\documentclass[11pt,a4paper]{article}

\usepackage{fullpage}
\usepackage[utf8x]{inputenc}
\usepackage{amsmath, amsthm}
\usepackage{wrapfig,graphicx,amssymb,textcomp,array,amsmath}
\usepackage{enumitem}
\usepackage{multirow}
\usepackage{tabularx}
\usepackage{algorithm}
\usepackage[usenames, dvipsnames]{color}
\usepackage[textsize=footnotesize]{todonotes}
\usepackage[titletoc,title]{appendix}
\usepackage{authblk}
\usepackage{wasysym}

\newcommand{\etal}{{et~al. }}

\title{Minimum Ply Covering of Points with Disks and Squares
}

\author{
	Therese Biedl\qquad
	Ahmad Biniaz\qquad 
	Anna Lubiw}
\affil{University of Waterloo, Canada}

\date{\today}
\newtheorem{lemma}{Lemma}

\newtheorem{theorem}{Theorem}

\newtheorem*{problem*}{Problem}
\newtheorem{claim}{Claim}
\newtheorem*{invariant*}{Invariant}

\begin{document}
	\maketitle
	\begin{abstract}
Following the seminal work of Erlebach and van Leeuwen in SODA 2008, we introduce the minimum ply covering problem. Given a set $P$ of points and a set $S$ of geometric objects, both in the plane, our goal is to find a subset $S'$ of $S$ that covers all points of $P$ while minimizing the maximum number of objects covering any point in the plane (not only points of $P$).
For objects that are unit squares and unit disks, this problem is NP-hard and cannot be approximated by a ratio smaller than 2. We present 2-approximation algorithms for this problem with respect to unit squares and unit disks. Our algorithms run in polynomial time when the optimum objective value is bounded by a constant.

Motivated by channel-assignment in wireless networks, 
we consider a variant of the problem where the selected unit disks must be {\em 3-colorable}, i.e., 
colored by three colors such that all disks of the same color are pairwise disjoint. We present a polynomial-time algorithm that achieves a 2-approximate solution, i.e., a solution that is 6-colorable.

We also study the weighted version of the problem in dimension one, where $P$ and $S$ are points and weighted intervals on a line, respectively. We present
an algorithm that solves this problem in $O(n + m + M )$-time where $n$ is the number of points, $m$ is the number of intervals, and $M$ is the number of pairs of overlapping intervals.
This repairs a solution claimed by Nandy, Pandit, and Roy in CCCG 2017.
	\end{abstract}
	
\section{Introduction}
Motivated by interference reduction in cellular networks, Kuhn \etal \cite{Kuhn2005} introduced  {\em Minimum Membership Set Cover} (MMSC) as a combinatorial optimization problem. The input to this problem consists of a set $U$ of elements and a collection $S$ of subsets of $U$, whose union contains all elements of $U$. The {\em membership} of an element $u\in U$ with respect to a subset $S'$ of $S$
is the number of sets in $S'$ that contain $u$. 
The goal is to find a subset $S'$ of $S$ that covers all elements of $U$ and that minimizes the maximum membership of elements in $U$.
The MMSC problem is closely related to the well-studied {\em Minimum Set Cover} problem
in which the goal is to find a minimum cardinality 
subset $S'$ of $S$ that covers all elements of $U$. By a reduction from the minimum set cover problem, Kuhn et al.~\cite{Kuhn2005} showed that the MMSC problem is NP-complete and cannot be approximated, in polynomial time, by a ratio less than $\ln n$ unless $NP\subset TIME\left(n^{O(\log \log n)}\right)$, where $n:=|U|$ is the number of elements. They also presented an $O(\ln n)$-approximation algorithm for the MMSC problem by formulating it as a linear program. 

Demaine \etal \cite{Demaine2008} introduced a maximization version of the MMSC problem, in which the input contains an extra parameter $\beta$, and the goal is to find a subset $S'$ of $S$ that covers the maximum number of elements of $U$ such that the membership of every covered element with respect to $S'$ is at most $\beta$. The special case, where $\beta=1$, is known as the {\em unique coverage} problem. See \cite{Demaine2008} for a collection of inapproximability results and \cite{Misra2013} for the parameterized complexity of the unique coverage problem.

The geometric versions of the above problems, where the elements are points in the plane and the sets are geometric objects in the plane, are also well studied. See \cite{Basappa2015,Biniaz2017,Mustafa2010} (and references therein) for some recent results on the geometric minimum set cover problem and its variants, and see \cite{Erlebach2008,Ito2014,Ito2016,vanLeeuwen2009} for some recent results on the geometric unique coverage problem.

The geometric MMSC problem attracted considerable attention following the seminal work of Erlebach and van~Leeuwen in SODA 2008 \cite{Erlebach2008}.
The input to this problem consists of a set $P$ of points and a set
$S$ of geometric objects both in the plane. The goal is to find a subset $S'$ of $S$ such that (i) $S'$ covers all points of $P$, i.e., the membership of every point is at least 1, and (ii) $S'$ minimizes the maximum membership of points of $P$. 
They proved that the geometric MMSC problem 
is NP-hard 
for unit disks and for axis-aligned unit squares, and does not admit a polynomial-time	approximation algorithm with ratio smaller than 2 unless P=NP. For unit squares, they presented a 5-approximation algorithm that takes polynomial time if the optimal objective value (i.e., the maximum membership) is bounded by a constant. To the best of our knowledge, no $O(1)$-approximation algorithm is known for unit disks.

\begin{wrapfigure}{r}{1.82in} 
	\vspace{-5pt} 
	\centering
	\includegraphics[width=.27\columnwidth]{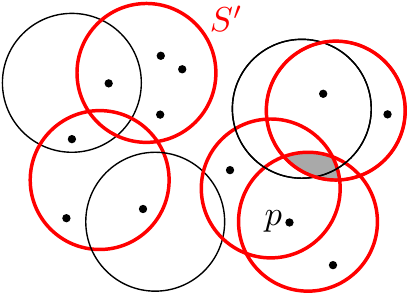}
	\vspace{-5pt} 
\end{wrapfigure}
In some applications, e.g. interference reduction in cellular networks, it 
is desirable to minimize 
the membership of every point in the plane, not only points of $P$. Thus we study a version of the geometric MMSC problem in which we want to find a subset $S'$ of $S$ that covers all points of $P$ and minimizes the maximum membership of all points of the plane (not only points of $P$). 
We refer to this version of the problem as {\em minimum ply covering} (MPC).
The {\em ply} \cite{EG08} of a set $S'$ is defined to be the maximum membership of points of the plane with respect to $S'$. With this definition, the MPC problem asks for a subset $S'$, with minimum ply, that covers $P$. In the figure to the right, the membership of input points with respect to $S'$ is at most 2 (see point $p$), while the ply of $S'$ is 3 (see the shaded area).

By a simple modification of the hardness proof of \cite{Erlebach2008}, we show (in Section~\ref{nphard-proof}) that the MPC problem is  
NP-hard, for both unit squares and unit disks, and does not admit polynomial-time approximation algorithms with ratio smaller than 2 unless P=NP.	
As our main result, we present 2-approximation algorithms for the MPC problem on unit squares and unit disks. Both algorithms run in polynomial time if the optimal objective value (i.e., the minimum ply) is bounded by a constant.

Motivated by channel-assignment in wireless networks, where the use of 3 channels is a standard practice, we study a variant of the MPC problem on unit disks where we want the solution to be 3-colorable, i.e., to be partitioned into
three subsets such that the disks in each subset are pairwise disjoint (each subset has ply 1). See \cite{Brass2010} for a justification of the importance of 3 channels. We present a polynomial-time 2-approximation algorithm for this version as well.

We also revisit the weighted version of the geometric MMSC problem in dimension one, where $P$ and $S$ are points and weighted intervals on the real line, respectively. This problem was previously claimed solved by Nandy, Pandit, and Roy 
\cite{Nandy2017}, who referred to the problem as ``minimum depth covering''. We point out a mistake in their algorithm.
We present an $O(n+m+M)$-time algorithm that solves this problem optimally, where $n$ is the number of points, $m$ is the number of intervals, and $M$ is the number of pairs of overlapping intervals. 
Our algorithm can be adapted in a simple way to solve the MPC problem on weighted intervals within the same time bound.

\section{Minimum Ply Covering with Unit Squares}
\label{square-section}
In this section we study the MPC problem on unit squares. We are given a set $P$ of $n$ points and a set $S$ of $m$ axis-aligned unit squares, both in the plane. We assume that unit squares are closed (contain their boundaries) and have side length 1. Our goal is to find a subset $S'$ of $S$, with minimum ply, that covers all points of $P$.
This problem cannot be approximated in polynomial-time by a ratio smaller than 2; see Section~\ref{nphard-proof}. 
We present a 2-approximation algorithm that takes polynomial time if the minimum ply is bounded by a constant. In the rest of this section we assume that the minimum ply is bounded by $\ell$.

We partition the plane into horizontal slabs of height 2; this is a standard initial step of many geometric covering algorithms. 
We may assume that no point of $P$ or edge of a square in $S$ lies on the boundary of any slab.  Let $H_1, H_2,\dots$ denote the slabs from bottom to top.
For $j\in\{1,2,\dots\}$, let $P_j$ be the points of $P$ in $H_j$ and let $S_j$ be the set of squares that intersect $H_j$. 
Note that if there exists a solution $S^*$ for the MPC problem
then $S^*\cap S_j$ covers all points in $P_j$ and has ply at most $\ell$,
and thus $S^*\cap S_j$ is a solution for the MPC problem on input instance $P_j$ and $S_j$ that has ply at most 
$\ell$.  Our approach is therefore to solve MPC for this input instance, i.e., for slab
$H_j$.  
If this fails for some $j$, then the MPC problem on $P$ and $S$ has no solution with ply $\ell$.  If 
this succeeds for all $j$, then we set $S'=\bigcup_j S_j'$, where $S_j'$ is the
solution for slab $H_j$.  
Certainly all points of $P$ are covered.
Any square in $S'$
belongs to solutions of at most two consecutive slabs $H_j$ and $H_{j+1}$.  Thus, any point in the plane is covered by squares of at most two solutions $S_j'$
and $S_{j+1}'$, and hence is covered by at most $2\ell$ squares of $S'$. Therefore, the ply of $S'$ is at most $2\ell$.

In the rest of this section we show how to solve the MPC problem for every slab $H_j$.  
To simplify our description we assume that the left and right sides of all squares in $S_j$ have
distinct $x$-coordinates and no point lies on the left or right side of a square; 
we will describe later how to handle coinciding $x$-coordinates. 
We partition the plane into vertical strips by lines through left and right sides of all squares in $S_j$. 
Let $t_0, t_1, \ldots, t_k$ denote these vertical strips, ordered from left to right. We consider every strip as an open set, i.e., the vertical
line between $t_{i-1}$ and $t_i$ belongs to neither of them. 
The leftmost strip $t_0$ is unbounded to the left and the rightmost strip $t_k$ is unbounded to the right. Since $|S_j|\!\leqslant\! m$, we have $k\! \leqslant\! 2m$. The following lemma is important for our strategy to solve the problem.

\begin{lemma}
	\label{continues-square-lemma}
	Let $S^*_j\!\subseteq\! S_j\!$ be any solution, with ply at most $\ell$, for the MPC problem. The number of squares in $S^*_j$ that intersect every strip $t_i$ is at most $3\ell$. 
\end{lemma}

\begin{wrapfigure}{r}{1.3in} 
	\vspace{-10pt} 
	\centering
	\includegraphics[width=.18\columnwidth]{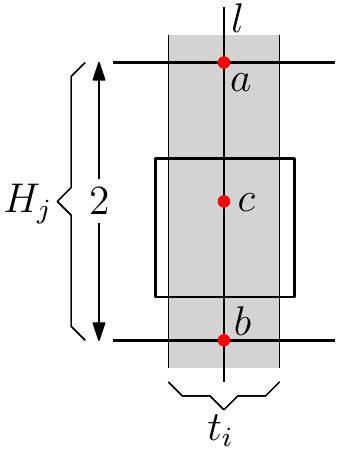}
	\vspace{-10pt} 
\end{wrapfigure}
\noindent{\em Proof.}
Let $l$ be a vertical line in the interior of $t_i$. Let $a$ and $b$ be the intersection points of $l$ with the upper and lower boundaries of the slab $H_j$, and let $c$ be the midpoint of the line segment $ab$; see the figure to the right. Notice that $|ac|=|bc|=1$. Because of this, and since no square has its left or right side in the interior of $t_i$, it follows that every square in $S^*_j$ that intersects $t_i$ contains at least one of the three points $a$, $b$ and $c$. Therefore, if more that $3\ell$ squares of $S^*_j$ intersect $t_i$, then by the pigeonhole principle one of the three points lies in more than $\ell$ squares of $S^*_j$, a contradiction.\qed
\vspace{10pt}

Based on Lemma~\ref{continues-square-lemma}, we construct a directed acyclic graph $G$ such that any solution $S^*_j$ corresponds to a path from the source vertex to the sink vertex in $G$. 
Afterwards, we will find a path in $G$ which will correspond to a solution with ply at most $\ell$. 
Now we describe the construction of $G$. 
For every $t_i$, with $i\in \{0,\dots,k\}$, we define a set $V_i$ of vertices as follows: 
For every subset $Q\subseteq S_j$, containing at most $3\ell$ squares that intersect $t_i$, we add a vertex $v_i(Q)$ to $V_i$ if the following conditions hold 

\begin{enumerate}
	\item[(i)] the squares in $Q$ cover all points of $P_j$ that lie in $t_i$ (all points in $t_i\cap P_j$),  
	\item[(ii)] the ply of $Q$ is at most $\ell$ (i.e.~every point in $\mathbb{R}^2$ is in at most $\ell$ squares of $Q$).
\end{enumerate}

Since no square intersects $t_0$ and $t_k$, we have $V_0=\{v_0(\emptyset)\}$ and $V_k=\{v_k(\emptyset)\}$. The vertices $v_0(\emptyset)$ and $v_k(\emptyset)$ are the source and sink vertices of $G$. The vertex set of $G$ is the union of the sets $V_i$. The edge set of $G$ consists of directed edges from the vertices in $V_i$ to the vertices in $V_{i+1}$ defined as follows. For every $i\in\{0,\dots,k-1\}$ and for every vertex $v_i(Q)\in V_i$ we add three directed edges from  $v_i(Q)$ to the following three vertices in $V_{i+1}$ (provided they exist):

\begin{enumerate}
	\item[1.] the vertex $v_{i+1}(Q')$ with $Q'=Q$,
	\item[2.] the vertex $v_{i+1}(Q')$ with $Q'=Q\setminus\{q\}$, where $q$ is the square whose right side is on the left boundary of $t_{i+1}$; see Figure~\ref{continues-square-fig}(a),	
	\item[3.] the vertex $v_{i+1}(Q')$ with $Q'=Q\cup \{q\}$, where $q$ is the square whose left side is on the left boundary of $t_{i+1}$; see Figure~\ref{continues-square-fig}(b).
\end{enumerate} 

This the end of our construction of $G$. Observe that in all cases sets $Q$ and $Q'$ differ by at most one square,
and $Q\subseteq Q'$ and/or $Q\supseteq Q'$.  

\begin{figure}[htb]
	\centering
	\vspace{-1pt}
	\setlength{\tabcolsep}{0in}
	$\begin{tabular}{cc}
	\multicolumn{1}{m{.55\columnwidth}}{\centering\includegraphics[width=.45\columnwidth]{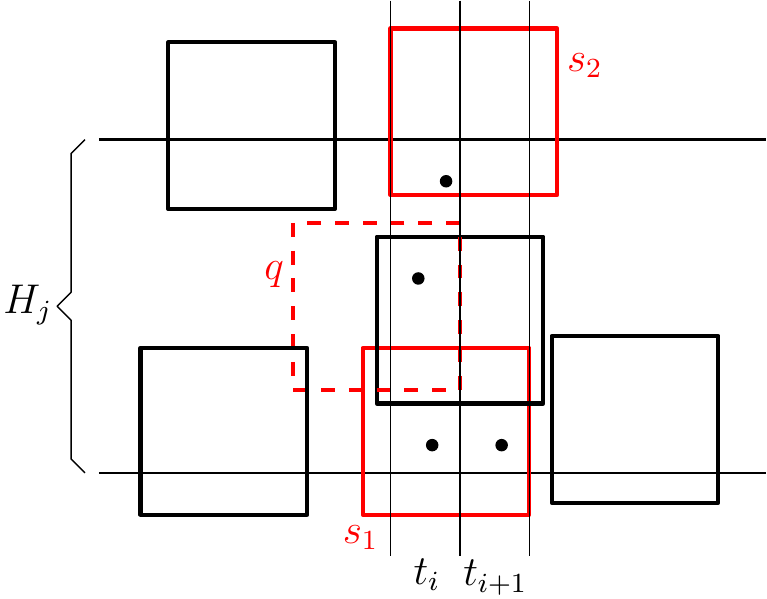}}
	&\multicolumn{1}{m{.45\columnwidth}}{\centering\includegraphics[width=.37\columnwidth]{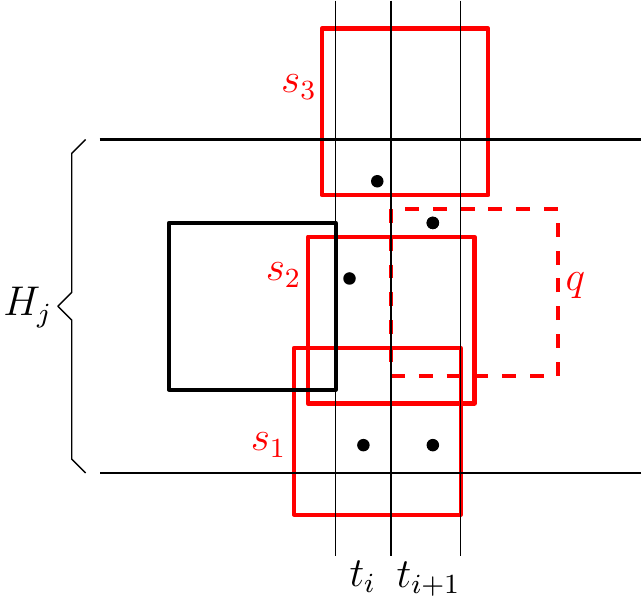}}\\
	(a)&(b)
	\end{tabular}$
	\vspace{-1pt}
	\caption{Construction of $G$. The representation of an edge from $v_i(Q)$ to $v_{i+1}(Q')$ in two cases: (a) $Q=\{s_1,s_2,q\}$ and $Q'=\{s_1,s_2\}$, and (b) $Q=\{s_1,s_2,s_3\}$ and $Q'=\{s_1,s_2,s_3,q\}$.}
	\vspace{-1pt}
	\label{continues-square-fig}
\end{figure}

Consider any path $\delta$ from $v_0(\emptyset)$ to $v_k(\emptyset)$ in $G$. Let $S'_j$ be the union of all sets $Q$ corresponding to the vertices of $\delta$. Our algorithm outputs $S'_j$ as a solution of the MPC problem on $P_j$ and $S_j$. The following claim proves the correctness of our algorithm.

\begin{claim}
	{If the MPC problem on $P_j$ and $S_j$ has a solution with ply at most $\ell$, then there exists a path $\delta$ from $v_0(\emptyset)$ to $v_k(\emptyset)$ in $G$. If there exist such a path $\delta$, then the set $S'_j$ is a solution with ply at most $\ell$.} 
\end{claim}
\begin{proof}
	For the first direction, consider a solution $S^*_j$ with ply at most $\ell$.
	For $0\leqslant i\leqslant k$, let $Q^*_i$ be the squares of $S^*_j$ that intersect $t_i$. 
	Observe that any point of $P_j$ that lies in $t_i$ is covered by a square that intersects $t_i$. Thus $Q^*_i$ covers all points in $t_i$.
	Since the ply of $S^*_j$ is at most $\ell$, the set $Q^*_i$ has at most $3\ell$ squares by Lemma~\ref{continues-square-lemma}. Therefore, $Q^*_i$ satisfies conditions (i) and (ii) and thus $v_i(Q^*_i)$ is a vertex of $G$.  Since (by our initial assumption) no two squares begin or end at the same $x$-coordinate, $Q^*_i$ and $Q^*_{i+1}$ differ by at most one square, 
	and thus there is an edge from $v_i(Q^*_i)$ to $v_{i+1}(Q^*_{i+1})$ in $G$. 
	Since this holds for every $i$, the solution $S^*_j$ can be
	mapped to the path
	$v_0(\emptyset), v_1(Q^*_1), \dots, v_{k-1}(Q^*_{k-1}), v_k(\emptyset)$ in $G$. Therefore, $\delta$ exists.
	
	For the other direction, observe that the edges of $G$ only connect the vertices of adjacent strips. Thus, $\delta$ contains exactly one vertex, say $v_i(Q_i)$, for each strip $t_i$, with $0\leqslant i\leqslant k$.  With this notation, we have that $S'_j=\bigcup_i Q_i$. By condition (i) the set $Q_i$ covers all points of $P_j$ that lie in $t_i$. 
	By our assumption of no coinciding $x$-coordinates, every point of $P_j$ lies in some (open) strip, and thus $S'_j$ covers all points of $P_j$. Now we verify the ply of $S'_j$.
	To that end, fix an arbitrary point $p\in \mathbb{R}^2$ and notice that $p$ can be in some strip $t_i$ or on the boundary between two strips $t_i$ and $t_{i+1}$. If $p$ is in $t_i$, then by condition (ii) it has membership at most $\ell$ in $Q_i$. Therefore it also has membership at most $\ell$ in $S'_j$ because by our definition of edges
	no square in $S'_j\setminus Q_i$ can intersect $t_i$. Assume now that  $p$ lies on the boundary between $t_i$ and $t_{i+1}$.  Any square of $S'_j$ that contains $p$ must belong to $Q_i$ or $Q_{i+1}$ (or both).
	Furthermore, by our construction of $G$, we have $Q_i\subseteq Q_{i+1}$ or $Q_i \supseteq Q_{i+1}$. Since $p$ 
	has membership at most $\ell$ in both $Q_i$ and $Q_{i+1}$ (by condition (ii)), the membership of $p$ in $S'_j$ is at most $\ell$. Therefore, the ply of $S'_j$ is at most $\ell$.
\end{proof}

\noindent{\bf Remark 1.}
Our algorithm does not use the fact that elements of $S$ are squares, but only uses that they have unit height.  Therefore the algorithm extends to axis-aligned unit-height rectangles. 

\vspace{8pt}
\begin{wrapfigure}{r}{2.3in} 
	\vspace{-10pt} 
	\centering
	\includegraphics[width=.33\columnwidth]{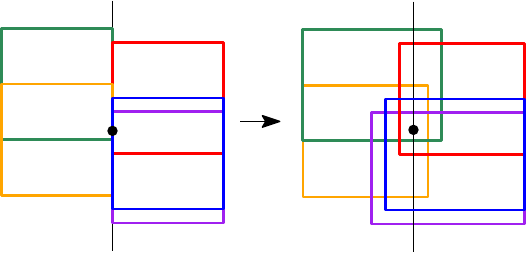}
	\vspace{-5pt} 
\end{wrapfigure}
\noindent{\bf Remark 2.}
The case where sides of squares and/or input points have coinciding $x$-coordinates can be handled 
by a symbolic perturbation.
At any $x$-coordinate $X$ we order first all the left sides of squares at $X$,
breaking ties by their $y$-ordering; then the input points at $X$; and finally,
all the right sides of squares at $X$, breaking ties by their $y$-ordering; see the figure to the right for illustration.
\vspace{8pt}

For the running time to solve the problem in $H_j$, set $n_j=|P_j|$ and $m_j=|S_j|$. Recall that every vertex of $G$ corresponds to a set $Q$ of at most $3\ell$ squares by our construction, and thus there are $O(m_j^{3\ell})$ such sets $Q$. This and the fact that each set $Q$ could be used repeatedly among $O(m_j)$ strips, imply that $G$ has $O(m_j^{3\ell+1})$ vertices. Since every vertex has at most three outgoing edges, the number of edges of $G$ is also $O(m_j^{3\ell+1})$.
By an initial sorting of the points of $P_j$ and the squares of $S_j$ with respect to the $y$-axis, conditions (i) and (ii) can be verified in $O(\ell+n_j)$ time for each vertex. 
A path $\delta$ can be found in time linear in the size of $G$. Thus, the total running time to solve the MPC problem in slab $H_j$ is $O((\ell +n_j) \cdot m_j^{3\ell+1})$.

\begin{theorem}
	\label{thm:height_2}
	There exists a polynomial-time algorithm that solves the problem of minimum ply covering of points in a slab of height two with unit-height rectangles, provided that the optimal objective value is constant.
\end{theorem}

As discussed at the beginning of this section, the union of the solutions of all slabs is a 2-approximate solution for the original problem of covering $n$ points in the plane with $m$ unit squares. 
Since every point belongs to exactly one slab and every square belongs to at most two slabs, this 2-approximate
solution can be computed in $\sum{(\ell+ n_j)  \cdot m_j^{3\ell+1}}=O((\ell+ n)  \cdot (2m)^{3\ell+1})$ time, where the sum runs over all slabs.
The following theorem summarizes our result.

\begin{theorem}
	There exists a polynomial-time 2-approximation algorithm that solves the problem of minimum ply covering of points with unit-height rectangles, provided that the optimal objective value is constant.
\end{theorem}

\section{Minimum Ply Covering with Unit Disks}
\label{disk-section}
In this section we study the MPC problem on unit disks, i.e., disks with diameter 1. Given a set
$P$ of $n$ points and a set $S$ of $m$ unit disks, both in the plane, the goal is to find a subset $S'$ of $S$, with minimum ply, that covers all points of $P$. 
This problem cannot be approximated in polynomial-time by a ratio better than 2; see Section~\ref{nphard-proof}. 
We present a 2-approximation algorithm
that takes polynomial time if the minimum ply is bounded by a constant. In the
rest of this section we assume that the minimum ply is bounded by $\ell$.
Our algorithm is a modification of that of unit squares. 
After a suitable rotation of the plane we assume that in the set consisting of points of $P$ together with the leftmost and rightmost points of disks in $S$, no two points have the same $x$-coordinate.

As in the previous section we partition the plane into horizontal slabs of height 2. Then we solve the MPC problem in every slab $H_j$ by constructing a directed acyclic graph $G$. Let $P_j$ be the set of points in $H_j$, and let $S_j$ be the set of all disks that intersect $H_j$. We partition the plane into vertical strips $t_0,\dots,t_k$ by vertical lines through the leftmost and rightmost points of disks in $S_j$. 
The only major change to the algorithm of the previous section is the definition of vertices of $G$, because Lemma~\ref{continues-square-lemma} does not hold for unit disks. For unit disks we have the following helper lemma.

\begin{lemma}
	\label{continues-disk-lemma}
	Let $S^*_j\subseteq S_j$ be any solution, with ply at most $\ell$, for the MPC problem. Then for any strip $t_i$ at most $8\ell$ disks in $S^*_j$ intersect $t_i$.
\end{lemma}

\begin{wrapfigure}{c}{1.9in} 
	\vspace{-8pt} 
	\centering
	\includegraphics[width=.27\columnwidth]{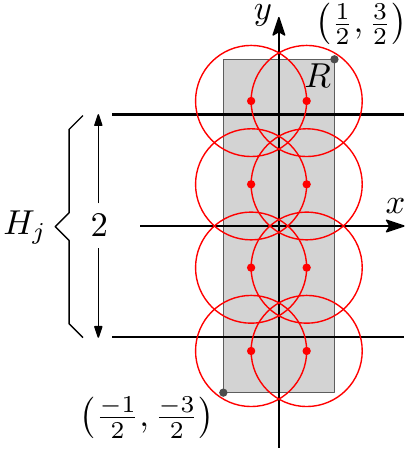}
	\vspace{-5pt} 
\end{wrapfigure}
\noindent{\em Proof.}
After a suitable translation assume that $H_j$ has $y$-range $[-1,+1]$, and assume that the $y$-axis lies in $t_i$.  Any disk $s\in S_j$, that intersects $t_i$, must also intersect the $y$-axis because boundaries of vertical strips are defined by vertical lines through leftmost and rightmost points of disks in $S_j$. It follows that the center of $s$ must lie within a rectangle $R$ with corners
$(\pm\frac{1}{2}, \pm\frac{3}{2}$); see the figure to the right.  The rectangle $R$ can be covered by eight unit disks $\{D_1,\dots,D_8\}$ that are centered at eight points $\{p_1,\dots,p_8\}=\{ (\pm\frac{1}{4},\pm\frac{3}{8}),(\pm\frac{1}{4},\pm\frac{9}{8})\}$ respectively---the red points in the figure.   
Thus, the center of $s$ lies in a disk $D_i$, for some $i\in\{1,\dots,8\}$, and hence at distance at most $\frac{1}{2}$ from $p_i$. This implies that $p_i\in s$.  Thus, each disk in $S_j$ that intersects $t_i$ contains at least one of the points $\{p_1,\dots,p_8\}$. Since $S^*_j$ has ply at most $\ell$, each point $p_i$ lies in at most $\ell$ disks of $S^*_j$. Therefore,
at most $8\ell$ disks of $S^*_j$ intersect $t_i$.\qed
\vspace{8pt}

For every strip $t_i$ we introduce a set $V_i$ that contains vertices $v_i(Q)$, for all sets $Q$ of at most $8\ell$ disks that (i) intersect $t_i$, (ii) cover all points of $P_j$ in $t_i$, and (iii) have ply at most $\ell$. Then we connect the vertices of $V_i$ to the vertices of $V_{i+1}$ in a similar fashion as for unit squares.

Using Lemma~\ref{continues-disk-lemma}, we can claim (similar to that of unit squares) that any path from the source to the sink in $G$ corresponds to a solution with ply at most $\ell$ for the MPC problem with input $P_j$ and $S_j$.   
Therefore, we can compute in $O((\ell +n_j) \cdot m_j^{8\ell+1})$ time a solution with ply at most $\ell$ for every slab, and in $O((\ell +n) \cdot (2m)^{8\ell+1})$ time a 2-approximate solution for the original problem.
The following theorem summarizes our result.

\begin{theorem}
	There exists a polynomial-time 2-approximation algorithm that solves the problem of minimum ply covering of points with unit disks, provided that the optimal objective value is constant.
\end{theorem}

\section{$3$-Colorable Unit Disk Covering}
\label{appB}

\begin{wrapfigure}{c}{1.2in} 
	\vspace{-8pt} 
	\centering
	\includegraphics[width=.18\columnwidth]{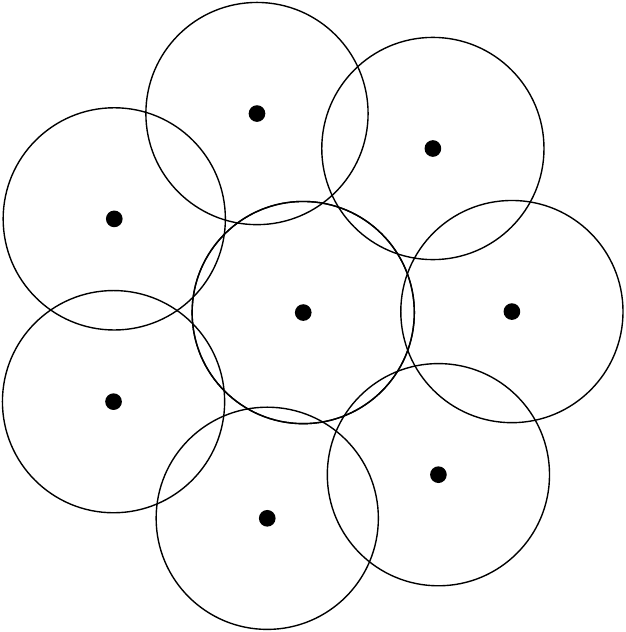}
	\vspace{-5pt} 
\end{wrapfigure}
In this section we study the 3-colorable unit disk covering problem. Given a set $P$ of $n$ points and a set $S$ of $m$ unit disks, both in the plane, the goal is to find a subset $S'$ of $S$ that covers all points of $P$ and such that $S'$ can be partitioned into $\{S'_1, S'_2,S'_3\}$ where the disks in each $S'_a$, $a\in\{1,2,3\}$, are pairwise disjoint, i.e., the ply of each $S'_a$ is 1. 
Although at first glance this problem seems to be a special case of the MPC problem with ply $\ell=3$,
they are different because there are input instances that have a solution with ply at most $3$ but do not have any 3-colorable solution; see for example the input instance in the figure to the right.
We present a polynomial-time algorithm that achieves a 2-approximate solution, i.e., a solution that is 6-colorable.

Before presenting our algorithm we point out a related problem that is the problem of coloring a unit disk graph (the
intersection graph of a set $S$ of unit disks in the plane) with $k$ colors.  This problem NP-hard for any
$k\geqslant  3$ \cite{Clark1990}. There are 3-approximation algorithms for this problem (see e.g. \cite{Graf1998,Peeters1991}), and a 2-approximation algorithm when the unit disk graph has constant clique number (see \cite[Chapter 4, Proposition 4.8]{Graf1995}). We note that this problem is also different from our 3-colorable unit disk cover problem; for
example if all the disks in $S$ have a common intersection and we place our point set $P$ in this intersection, then there
exists a 1-colorable solution for our problem, while the unit disk graph is not
$(|S| - 1)$-colorable.
The NP-hardness of 3-coloring a unit disk graph \cite{Clark1990} immediately
implies the NP-hardness of the 3-colorable unit disk cover problem: every disk in the reduction of \cite[Theorem 2.1]{Clark1990} covers a unique region of the plane, so by placing points of $P$ in those regions we can enforce the solution $S'$ to contain all disks of $S$.  

Now we present our algorithm, which is a modified version of the algorithm of Section \ref{disk-section}. 
Again we partition the plane into horizontal slabs of height 2. Then for every slab $H_j$ we test the existence of a 3-colorable covering for points $P_j$ using disks in $S_j$. If this test fails for some $j$, then there is no 3-colorable solution for $P$ and $S$. If this test is successful for all $j$, then we assign colors 1, 2, 3 to solutions of $H_1, H_3,\dots$ and assign colors 4, 5, 6 to solutions of $H_2,H_4,\dots$. The union of these solutions will be a 6-colorable solution for the original problem. 

After a suitable rotation we assume that in the set  consisting of  points of $P$ together with the leftmost and rightmost points of disks in $S$, no two points have the same $x$-coordinate.
To solve the problem for every $H_j$, as in Section \ref{disk-section}, we partition the plane into vertical strips $t_0,\dots,t_k$. Then we construct a directed acyclic graph $G$ such that any path from the source to the sink in $G$ corresponds to a 3-colorable solution for $H_j$.   Consider a 3-colorable solution $S^*=S_1^*\cup S_2^* \cup S_3^*$. The disks in each $S_a^*$ (for $a=1,2,3$) are pairwise disjoint, and thus each $S_a^*$ has ply 1. Therefore, by Lemma~\ref{continues-disk-lemma} at most $8$ disks in each $S^*_a$ intersect each strip $t_i$.
Based on this, for every $t_i$ we introduce a set $V_i$ containing vertices $v_i(Q_1,Q_2,Q_3)$ for all sets $Q_1$, $Q_2$, and $Q_3$ that satisfy all following conditions:

\begin{enumerate}[leftmargin=2.5em]
	\item[(i)] each of $Q_1$, $Q_2$, and $Q_3$ contains at most $8$ disks that intersect $t_i$,
	\item[(ii)] the disks in each of $Q_1$, $Q_2$, and $Q_3$ are pairwise disjoint, and
	\item[(iii)] $Q_1\cup Q_2\cup Q_3$ covers all points of $P_j$ that lie in $t_i$.
\end{enumerate}

We connect a vertex $v_i(Q_1,Q_2,Q_3)$ to a vertex $v_{i+1}(Q'_1,Q'_2,Q'_3)$ if one of the following conditions hold:

\begin{itemize}
	\item for every index $a$ in $\{1,2,3\}$ we have $Q'_a=Q_a$, or 
	\item for exactly one index $a$ in $\{1,2,3\}$ we have $Q'_a=Q_a\setminus\{d\}$, where $d$ is the disk whose rightmost point is on the
	left boundary of $t_{i+1}$, and for every other index $b$ we have $Q'_b=Q_b$, or
	\item for exactly one index $a$ in $\{1,2,3\}$ we have $Q'_a=Q_a\cup\{d\}$, where $d$ is the disk whose leftmost point is on the
	left boundary of $t_{i+1}$, and for every other index $b$ we have $Q'_b=Q_b$.
\end{itemize}

We briefly justify that paths from the source to the sink in $G$ correspond to 3-colorable solutions.   Fix a path $\delta$.
For $a=1,2,3$, let $S'_a$ be the union of the disks in all sets $Q_a$ associated with vertices in $\delta$. Condition (iii) ensures that $S'_1\cup S'_2\cup S'_3$ covers all points of $P_j$. Condition (ii) ensures that the disks in each of $S'_1$, $S'_2$, and $S'_3$ are pairwise disjoint. Thus, $S'_1\cup S'_2\cup S'_3$ is a 3-colorable solution. 
It remains to verify the existence of such a path $\delta$ in $G$. Consider a 3-colorable solution $S^*=S_1^* \cup S_2^* \cup S_3^*$. For any $a\in \{1,2,3\}$ and $i\in\{0,\dots,k\}$ let $Q^i_a$ be the disks in
$S_a^*$ that intersect $t_i$.  As discussed above, we have $|Q^i_a|\leqslant 8$.
Thus, $v_i(Q^i_1,Q^i_2,Q^i_3)$ is a vertex of $G$. By an argument similar to that of Section~\ref{square-section}, one can verify that $S^*$ can be mapped to a path in $G$. Therefore, $\delta$ exists.

The running time analysis is similar to that of Section \ref{disk-section}, except here we consider 24 disks in every vertical strip (8 disks for every set $Q_i$). 
Therefore the total running time of our 2-approximation algorithm is $O(n m^{25})$.
The following theorem summarizes our result.

\begin{theorem}
	\label{three-color-thr}
	There exists a polynomial-time 2-approximation algorithm for the 3-colorable covering problem of points in the plane with unit disks.
\end{theorem}

\section{Minimum Membership and Minimum Ply Coverings with Weighted Intervals}

Nandy, Pandit, and Roy \cite{Nandy2017} studied the geometric MMSC problem in dimension one, where $P$ is a set of $n$ points and $S$ is a set of $m$ closed intervals, both on the real line. 
By formulating this problem as an instance of the maximum independent set problem, they solved it in $O(n+m)$ time, provided that the points of $P$ and the endpoints of intervals in $S$ are given in sorted order. 

Nandy \etal \cite{Nandy2017} also studied the weighted version of the geometric MMSC problem in dimension one. In this version every interval has a weight, and the membership of a point $p$ is defined as the sum of the weights of the intervals in $S'$ that cover $p$. They claimed an $O(nm\log n)$-time algorithm that solves this version of the problem. 
However, as we point out in Section \ref{Nandy}, their algorithm does not always find the optimal solution. 
We present here an algorithm that solves the weighted version of the MMSC problem in dimension one in $O(n + m + M )$ time, where $M\leqslant {m\choose 2}$ is the number of pairs 
of overlapping intervals, i.e., the number of edges of the interval graph formed by the intervals in $S$.
As we will see later in Remark 3, our algorithm can easily be modified to also solve the MPC problem on weighted intervals in the same time.

\subsection{Nandy, Pandit, and Roy's Algorithm}
\label{Nandy}
We briefly review the algorithm of Nandy \etal \cite{Nandy2017},
which uses dynamic programming and proceeds as follows.		
Let $p_1,p_2,\dots, p_n$ be the points of $P$ from left to right, and let $s_1, s_2, \dots, s_m$ be the intervals in $S$ sorted from left to right according to their right endpoints. (In the original presentation of this algorithm in \cite{Nandy2017}, the points of $P$ are ordered from right to left and the intervals in $S$ are also sorted from right to left according to their left endpoints.) 
Let $w(s_i)$ denote the weight of $s_i$. 
For every $k\in\{1,\dots,m\}$ let $Q_k = (P_k, S_k)$ be an instance of the problem, where $S_k=\{s_1,\dots,s_k\}$ and $P_k$ is the set of points of $P$ that lie on or to the left of the right endpoint of $s_k$. Note that $Q_m$ is the original problem.  Now solve $Q_1,\dots,Q_m$ in order, i.e., process interval $s_k$ to find a solution $S'_k$ of $Q_k$, using previously computed 
solutions $S'_1,\dots,S'_{k-1}$ for subproblems $Q_1,\dots,Q_{k-1}$. 
Nandy \etal claim that the following approach will find $S'_k$:

\begin{enumerate}[leftmargin=2.5em]
	\item[(i)] If $P_k=P_{k-1}$ then set $S'_k=S'_{k-1}$.
	\item[(ii)] If $P_k\neq P_{k-1}$ and some points of $P_k\setminus P_{k-1}$ are not covered by $s_k$, then there is no solution for $Q_k$, i.e., $S'_k$ does not exist.
	\item[(iii)] If $P_k\neq P_{k-1}$ and every point of $P_k\setminus P_{k-1}$ is covered by $s_k$, then find the index $i$, with $i\in\{1,\dots,k-1\}$, such that $S'_i\cup\{s_k\}$ covers $P_k$ and the membership of points of $P_k$ with respect to $S'_i\cup\{s_k\}$ is minimum. Then set $S'_k=S'_i\cup\{s_k\}$.
\end{enumerate}

There are two reasons why this algorithm is not correct.  First, 
item (i) does not consider $s_k$, while in some cases we may need to include it in the solution, for example if $s_k$ has the minimum weight and covers all points in $P_k$.  
This may perhaps be fixable with a minor change of formula, but there is major issue in item (iii) which somehow breaks the hope for any straightforward dynamic programming approach for this problem.

\begin{figure}[htb]
	\vspace{-1pt}
	\centering
	\includegraphics[width=.6\columnwidth]{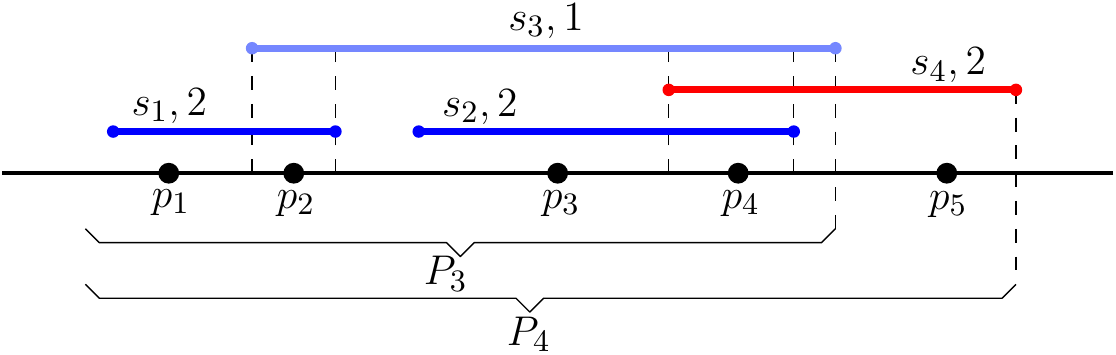}
	\vspace{-1pt}
	\caption{An instance for which the previous algorithm does not compute an optimal solution.  ``$s_1,2$'' indicates that interval $s_1$ has weight 2.}
	\vspace{-1pt}
	\label{counterexample-fig}
\end{figure}

Figure~\ref{counterexample-fig} shows an example with five points $p_1,p_2,p_3,p_4,p_5$, ordered from left to right, and four weighted intervals $s_1,s_2,s_3,s_4$, ordered according to their right endpoints. 
The optimal solutions for $Q_1$, $Q_2$, $Q_3$ are $S'_1=\{s_1\}$, $S'_2=\{s_1,s_2\}$, $S'_3=\{s_1,s_2\}$, with maximum memberships 2, 2, 2. Notice that $S'_3$ should not contain $s_3$ because otherwise it should also contain $s_1$ to cover $p_1$, and this will make the membership of $p_2$ be $w(s_1)+w(s_3)=3$. The recursive computation in item (iii) considers as solution for $Q_4$ the candidates $S'_1\cup \{s_4\}$, $S'_2\cup \{s_4\}$, $S'_3\cup \{s_4\}$. The first one is not valid (it does not cover $p_3$) and the other two have maximum membership $w(s_2)+w(s_4)=4$ (this membership is obtained by $p_4$ which is covered by $s_2$ and $s_4$).  However, the set $S'_4=\{s_1,s_3,s_4\}$ is a solution for $Q_4$ with maximum membership 3. We are skeptical about using a straightforward dynamic programming
for solving the problem, because while $S'_4$ is the optimum solution for $Q_4$, its induced solutions for $Q_2$ and $Q_3$ contain the superfluous interval $s_3$. To use the dynamic	programming approach, a more structured recursion is needed, as described in the next section.

\subsection{Our Algorithm}
\label{algorithm-section}

Now we present our algorithm for the geometric MMSC problem on weighted intervals. 
Our approach is to create a directed acyclic graph (DAG) such that all optimal solutions correspond to directed paths from the source to the sink. More precisely, we construct a vertex-weighted DAG (because the intervals are weighted), and then search for a {\em bottleneck path} from the source to the sink; a bottleneck path is a path that minimizes the maximum weight along the path. Such a DAG could be obtained directly from Section~\ref{square-section}, but we introduce a new construction with better running time.

Let $P$ be a set of $n$ points and let $S=\{s_1,\dots,s_m\}$ be a set of $m$ intervals on a line, where each $s_i$ has weight $w(s_i)$. 
Assume that the points of $P$ and the endpoints of intervals in $S$ are given in sorted order from left to right. In the remainder of this section we show how to transform our problem, in $O(n + m + M )$ time, to an instance of the bottleneck path problem in a DAG of size $O(m + M)$. The bottleneck path problem in a DAG can be solved in linear time.
For simplicity of our description we assume points of $P$ and endpoints of intervals of $S$ have distinct $x$-coordinates; this can be achieved by a symbolic perturbation similar to that of Remark 2 in Section~\ref{square-section}.

\begin{figure}[h]
	\centering
	\includegraphics[width=.999\columnwidth]{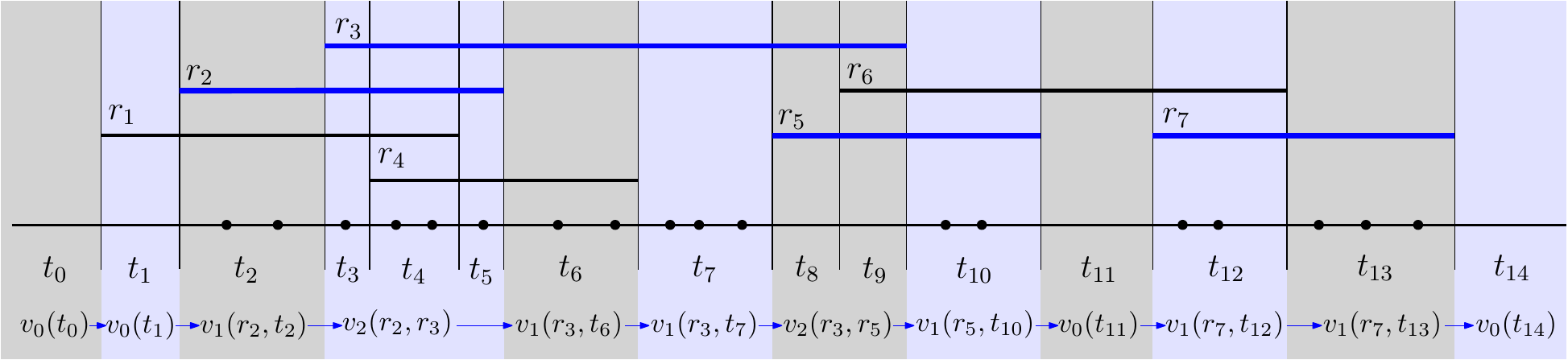}
	\vspace{-1pt}
	\caption{Correspondence between an optimal solution ($\{r_2,r_3,r_5,r_7\}$) and a bottleneck path from $v_0(t_0)$ to $v_0(t_{14})$ in $G$.}
	\vspace{-1pt}
	\label{slabs-fig}
\end{figure}

Draw vertical lines through the endpoints of intervals in $S$ to
partition the plane into vertical strips $t_0,\dots,t_{k}$ (ordered from left to right).
Notice that $k \leqslant 2m$, the leftmost strip $t_0$ is unbounded to the left, and the rightmost strip $t_k$ is unbounded to the right. 
Also $t_0$ and $t_k$ have no points in them. See Figure~\ref{slabs-fig}.
For every $i\in\{1,\dots, k\}$, all points of $P$ in strip $t_i$ belong to the same set of intervals. Thus, it suffices to keep only one of the points in each strip $t_i$ and discard the other points. 
This can be done in $O(n+m)$ time. Thus, we assume that $n\leqslant 2m-1$. 
We say that an interval $s$ is {\em dominated} by an interval $r$ if $r$ starts before $s$ and ends after $s$. The following lemma has been proved in \cite{Nandy2017}.

\begin{lemma}
	\label{two-segment-lemma}
	There exists an optimal solution in which no interval is dominated by some other interval, and no vertical line intersects more than two intervals.
\end{lemma}

Based on this lemma, it suffices to search for 	
an optimal solution $S^*\subseteq S$ such that the interior of every strip $t_i$ intersects at most two intervals of $S^*$; see Figure~\ref{slabs-fig}. 
We therefore endow our DAG, denoted by $G$, with the following three types of vertices, corresponding to the intersection of strips by 0, 1, or 2 intervals, respectively.

\begin{itemize}[leftmargin=1.5em]
	\item For every strip $t_i$ that contains no points of $P$, create a vertex $v_0(t_i)$ with weight 0.  
	Choosing $v_0(t_i)$ in our path will correspond to having 0 intervals cover the interior of $t_i$.  
	Vertices $v_0(t_0)$ and $v_0(t_k)$ will be the source and sink vertices of our graph. 
	
	\item For every interval $q$ and for every strip $t_i$ that is intersected by $q$, create a vertex $v_1(q,t_i)$. Choosing $v_1(q,t_i)$ in our path will correspond to choosing interval $q$ and choosing no other interval that intersects $t_i$. If some point of $P$ is in $t_i$, then set $w(v_1(q,t_i)) = w(q)$, otherwise set $w(v_1(q,t_i)) = 0$.
	
	\item For every two overlapping intervals $q$ and $r$, with $q$ starting before $r$ starts and also ending before $r$ ends, create a vertex $v_2(q,r)$. Choosing $v_2(q,r)$ in our path will correspond to choosing intervals $q$ and $r$.  
	If there is a point of $P$ in the intersection of $q$ and $r$, then set $w(v_2(q,r)) = w(q) + w(r)$, otherwise set $w(v_2(q,r))=0$.
\end{itemize}

We define edges as follows.  Consider each $i<k$ and the strip $t_i$.  The right boundary of $t_i$ exists since some interval has an endpoint there.  Assume first that this was a left endpoint, so there exists an interval $r$ whose left endpoint lies on the right boundary of $t_i$ (and $r$ is unique by our assumption). Then we add the edges listed in Figure~\ref{graph-fig}(a), adding only those where the vertices exist. 
If instead some interval $s$ has its right endpoint at the right boundary of $t_i$, then we add the following edges listed in Figure~\ref{graph-fig}(b) (again, provided the vertices listed there exist). Note that every vertex has at most two outgoing edges.  

\begin{figure}[htb]
	\centering
	\setlength{\tabcolsep}{0in}
	$\begin{tabular}{cc}
	\multicolumn{1}{m{.5\columnwidth}}{\centering\includegraphics[width=.23\columnwidth]{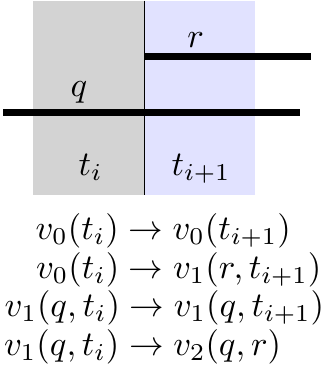}}
	&\multicolumn{1}{m{.5\columnwidth}}{\centering\includegraphics[width=.23\columnwidth]{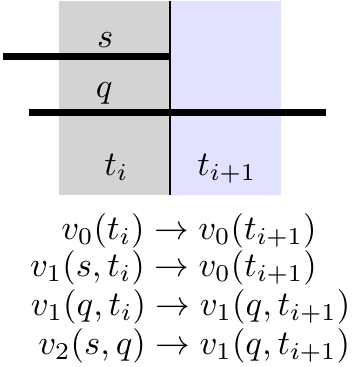}}
	\\(a)&(b)
	\end{tabular}$
	\vspace{-1pt}
	\caption{The edges of $G$.} 
	\vspace{-1pt}
	\label{graph-fig}
\end{figure}

The above construction of graph $G$ implies the following claim.

\begin{claim}
	{The solutions that satisfy the constraints of Lemma~\ref{two-segment-lemma}
		are in one-to-one correspondence with the paths from $v_0(t_0)$ to $v_0(t_k)$ in $G$.  In particular, the bottleneck path in $G$ corresponds to an optimal solution of the minimum membership covering problem for $P$ and $S$. See Figure~\ref{slabs-fig} for an illustration.}
\end{claim}

Now we analyze the running time of our algorithm.  Since $G$ is a DAG, a bottleneck path in $G$ can be computed in linear time on the size of $G$.

We claim that the size of $G$ is $O(m+M)$, where $M$ is number of edges in the interval graph defined by $S$.  Since every vertex has
at most two outgoing edges, it suffices to show that $G$ has $O(m+M)$ vertices.  There are at most $2m+1$ vertices of type
$v_0$ (one per strip).  
Every vertex of type $v_2$ is uniquely defined by two overlapping intervals, so there are at most $M$ of them. 
To bound the number of vertices of type $v_1$, observe that such a vertex is defined by an interval $r$ and a strip $t_i$
with $1\leqslant i < k$.  Let $r'$ be the interval that defines the left boundary of $t_i$; it might be the case that $r'=r$.  The number of vertices for which $r'=r$ is at most $2m-1$ because this happens (by the general position assumption) only once per strip.  Assume that $r'\neq r$. Since $r$ and $r'$  intersect (at the left boundary of $t_i$), it follows that $(r,r')$ is an edge in the interval graph.   Any such edge can be charged 
at most 4 times, once at each of the left and the right boundaries of $r$ and $r'$. Thus, the number of vertices for which $r'\neq r$ is at most $4M$.
Hence $G$ has size $O(m+M)$ and the bottleneck path can be found within the same time.

The linked-list representation of $G$ can be obtained in $O(m+M)$ time by going through the strips from left to right.
This, together with the initial pruning of points of $P$ in $O(n+m)$ time, implies that the total running time of our algorithm is $O(n+m+M)$. 

\vspace{8pt}
\noindent{\bf Remark 3.}
We can solve the MPC problem on $n$ points and $m$ weighted intervals in $O(n+m+M)$ time; this can be done by adjusting the weights of the vertices of $G$ and then finding a bottleneck path in the resulting graph. We set the weight of every vertex $v_1(q,t_i)$ to $w(q)$ regardless of whether $t_i$ contains a point of $P$ or not, and we set the weight of every vertex $v_2(q,r)$ to $w(q)+w(r)$ regardless of whether the intersection of $q$ and $r$ contains a point of $P$ or not. 


\begin{theorem}
	Both minimum membership and minimum ply covering problems of $n$ points with $m$ weighted intervals on the real line can be solved in $O(n+m+M)$ time where $M$ is the number of pairs of overlapping intervals, provided that the input points and the endpoints of the intervals are given in sorted order.
\end{theorem}

\section{NP-hardness}
\label{nphard-proof}

In this section we study the hardness of the MPC problem on unit squares and unit disks. Erlebach and van Leeuwen \cite{Erlebach2008} proved that the MMSC problem, on both unit disks and unit squares, is NP-hard and cannot be approximated by a ratio smaller than 2. More precisely, they proved it is NP-hard to decide whether or not a solution with membership 1 exists. They use a reduction from the NP-complete problem of deciding whether or not a planar graph $G$ is 3-colorable. Their reduction uses a rectilinear embedding of $G$ in the plane; see \cite[Chapter 10]{vanLeeuwen2009} for details. By adjusting the gadgets in their reduction we can show that the MPC problem on unit squares and unit disks, with objective ply 1, is also NP-hard and cannot be approximated by a ratio smaller than 2. We sketch the reduction using our adjusted gadgets, and hence 
prove the following theorem. 

\begin{theorem}
	\label{hardness-thr}
	It is NP-hard to decide whether or not there exists a solution with ply one, for the minimum ply cover problem on unit squares and on unit disks. 
\end{theorem}

Since ply---of a set of squares or disks---is an integer, this theorem implies that there is no polynomial-time approximation algorithm, with ratio smaller than 2, for the MPC problem on unit squares and on unit disks, unless P=NP.

\begin{figure}[h]
	\centering
	\setlength{\tabcolsep}{0in}
	$\begin{tabular}{cc}
	\multicolumn{1}{m{.37\columnwidth}}{\centering\includegraphics[width=.16\columnwidth]{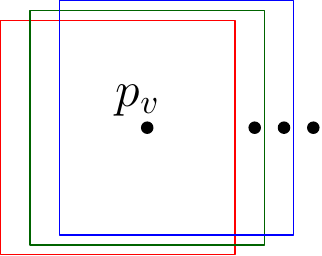}}
	&\multicolumn{1}{m{.63\columnwidth}}{\centering\includegraphics[width=.47\columnwidth]{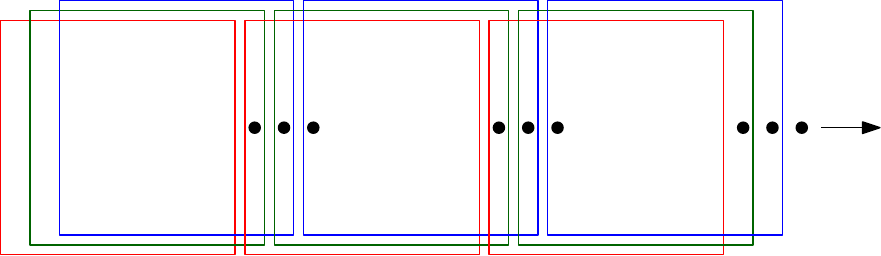}}\\
	(a)&(b)\\[.3cm]
	\multicolumn{1}{m{.37\columnwidth}}{\centering\includegraphics[width=.33\columnwidth]{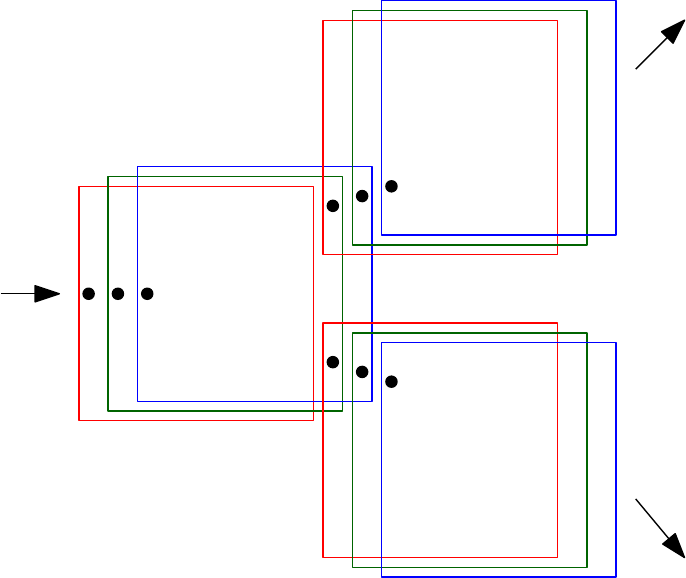}}
	&\multicolumn{1}{m{.63\columnwidth}}{\centering\includegraphics[width=.6\columnwidth]{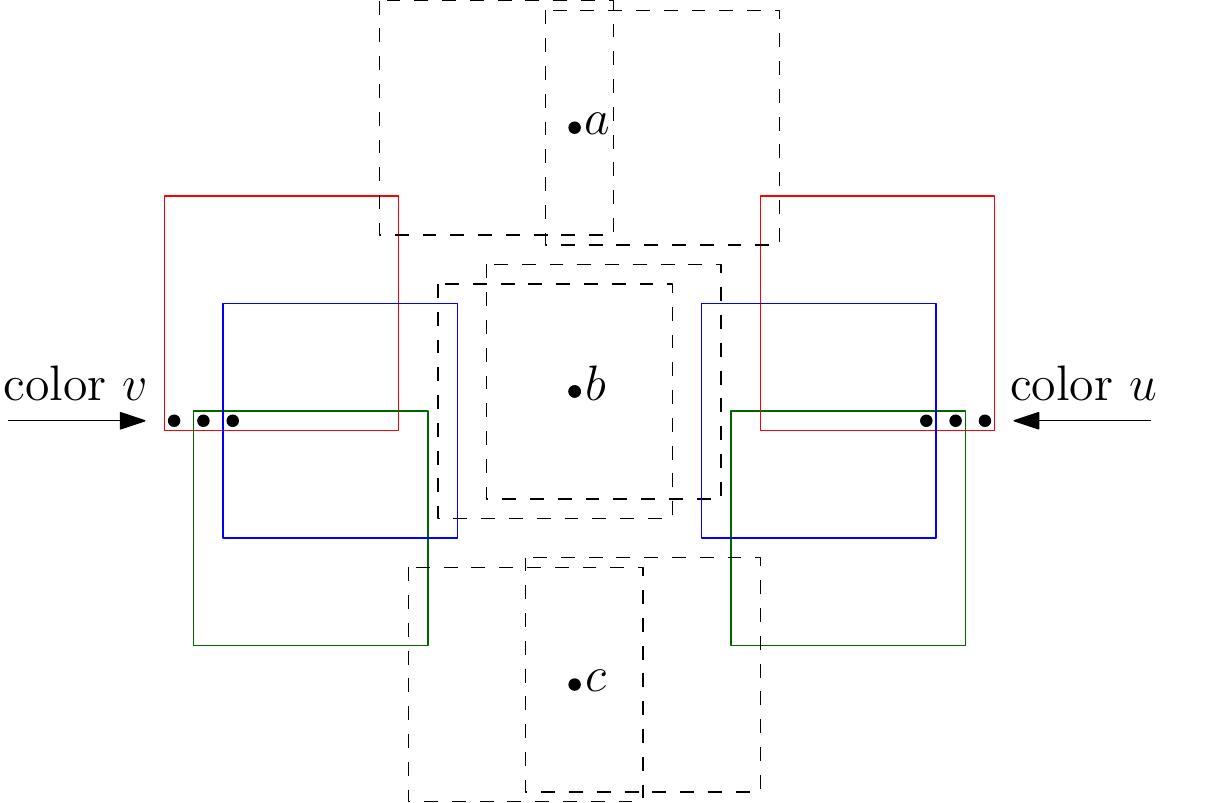}}\\
	(c)&(d)
	\end{tabular}$
	\vspace{-1pt}
	\caption{The NP-hardness of the MPC problem on unit squares: (a) the vertex gadget, (b) the transport gadget, (c) the duplicate gadget, and (d) the color checking gadget.}
	\label{continues-square-hardness-fig}
\end{figure}

In the rest of this section we prove Theorem~\ref{hardness-thr}. Similar to that of Erlebach and van Leeuwen \cite{Erlebach2008}, our proof uses a reduction from the problem of 3-coloring a planar graph $G$. We sketch this reduction for the MPC problem on unit sqaures; the reduction for unit disks is analogous. 
We need a few gadgets, their construction is described below.

Figure~\ref{continues-square-hardness-fig}(a) shows the vertex gadget for every vertex $v$ in $G$. To cover the point $p_v$, exactly one of the three squares containing $p_v$ must be in the solution, and this corresponds to assigning
one of the three colors to $v$. Depending on this choice, either 0, 1, or 2 points among the triple of points on the right will also be covered. Figure~\ref{continues-square-hardness-fig}(b) shows a transport gadget, which transports a chosen
color along a chain of squares from left to right. The reader may verify that this representation can be modified to bend around corners to represent vertical transportations. Figure~\ref{continues-square-hardness-fig}(c) shows a gadget that duplicates a chosen color. We need one extra gadget to make sure that two adjacent vertices $u$ and $v$ in $G$ will be assigned different colors. Figure~\ref{continues-square-hardness-fig}(d) shows this gadget, assuming the color of $u$ arrives from right and the color of $v$ arrives from left. If $u$ and $v$ have the same color, then to cover $a$, $b$ and $c$, we require overlapping squares, which contradicts the ply being 1. If $u$ and $v$ have different colors, then $a$, $b$, and $c$ can be covered by three of the six dashed disks. Therefore, our instance of the MPC problem has a solution with ply 1 if and only if $G$ is 3-colorable. Thus the hardness of the 3-colorability problem implies the hardness of the MPC problem. A similar construction of gadgets for the NP-harness of the MPC problem on unit disks is shown in Figure~\ref{continues-disk-hardness-fig}; the points $b$ and $c$ play the same role as in Figure~\ref{continues-square-hardness-fig}(d), while $a_1, a_2,a_3$ play the role of $a$. Therefore, the above reduction proves Theorem~\ref{hardness-thr}.

\begin{figure}[h]
	\centering
	\setlength{\tabcolsep}{0in}
	$\begin{tabular}{ccc}
	\multicolumn{1}{m{.18\columnwidth}}{\centering\includegraphics[width=.13\columnwidth]{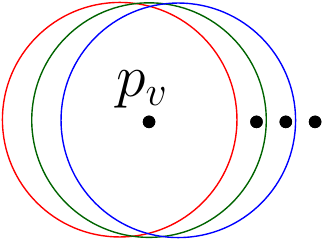}}
	&\multicolumn{1}{m{.47\columnwidth}}{\centering\includegraphics[width=.4\columnwidth]{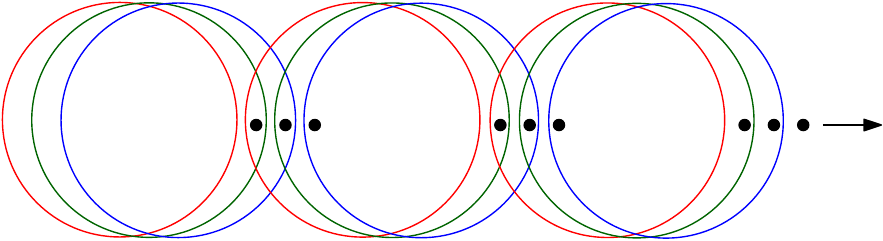}}
	&\multicolumn{1}{m{.35\columnwidth}}{\centering\includegraphics[width=.27\columnwidth]{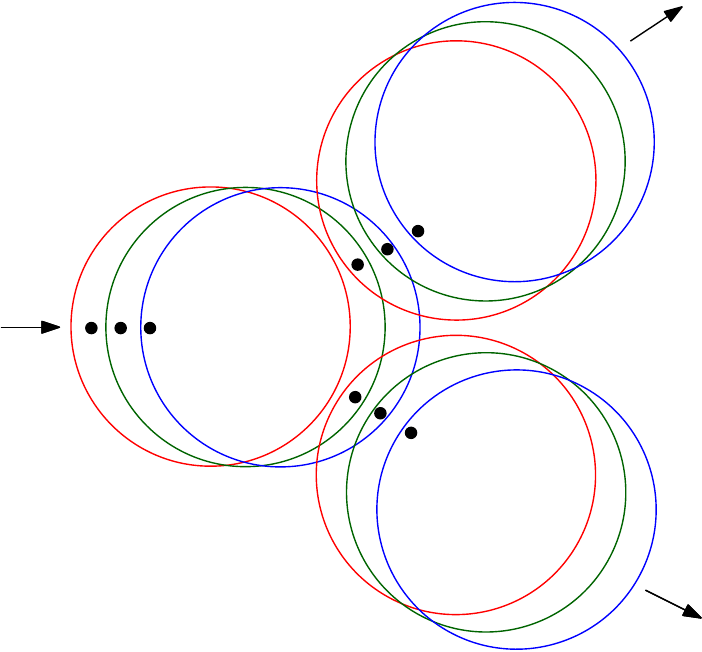}}\\
	\multicolumn{3}{m{.99\columnwidth}}{\centering\includegraphics[width=.77\columnwidth]{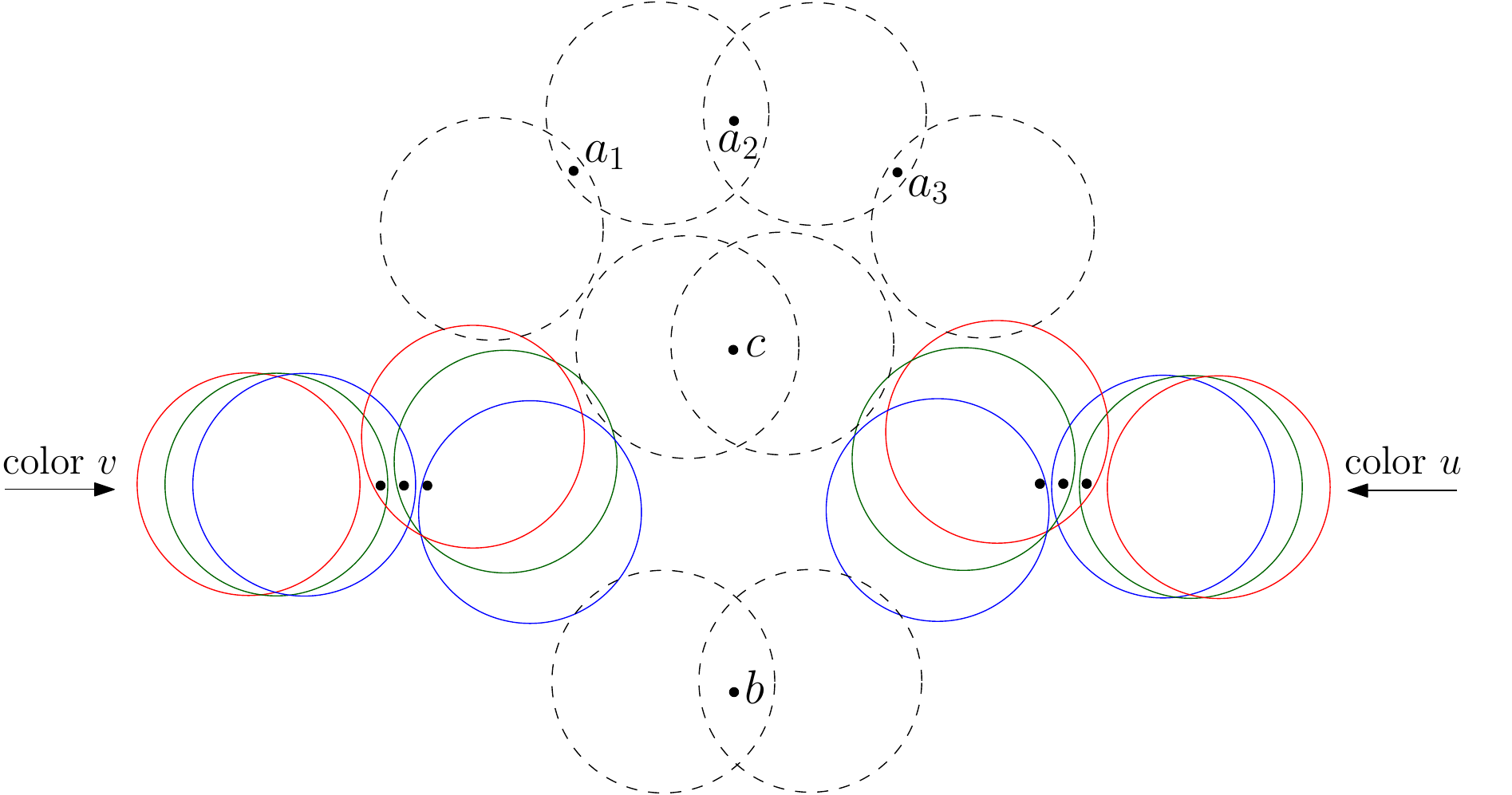}}
	\end{tabular}$
	\vspace{-1pt}
	\caption{Gadgets for the NP-hardness proof of the MPC problem on unit disks.}
	\label{continues-disk-hardness-fig}
\end{figure}

\section{Conclusions}
The following questions arise from this work and that of Erlebach and van Leeuwen \cite{Erlebach2008}: (i) Are there polynomial-time $O(1)$-approximation algorithms for the MPC problem on unit squares or unit disks when the objective value is not necessarily a constant? (ii) Is there a polynomial-time $O(1)$-approximation algorithm for the MMSC problem on unit disks when the objective value is constant?  

\paragraph{Acknowledgment.}
We thank Albert Gr\"{a}f for sending us his PhD thesis.

\bibliographystyle{abbrv}
\bibliography{Minimum-Depth-Covering.bib}

\end{document}